\documentclass{amsart}
\usepackage{colortbl}
\usepackage{amsmath,amsthm,amssymb}
\usepackage{graphicx}
\usepackage{harvard}

\renewcommand{\L}{\ensuremath{\mathcal{L}}}

\newtheorem{Definition}{Definition}

\newtheorem{Theorem}[Definition]{Theorem}

\begin{document}
\title{Moment Methods for Exotic Volatility Derivatives}

\author{Claudio Albanese}
\address{Independent consultant, level3finance.com}
\email{claudioalbanese@level3finance.net}

\author{Adel Osseiran}
\address{Mathematics, Imperial College London.}
\email{a.osseiran@imperial.ac.uk}

\date{\today}

\begin{abstract}
The latest generation of volatility derivatives goes beyond variance
and volatility swaps and probes our ability to price realized
variance and sojourn times along bridges for the underlying stock
price process. In this paper, we give an operator algebraic
treatment of this problem based on Dyson expansions and moment
methods and discuss applications to exotic volatility derivatives.
The methods are quite flexible and allow for a specification of the
underlying process which is semi-parametric or even non-parametric,
including state-dependent local volatility, jumps, stochastic
volatility and regime switching. We find that volatility derivatives
are particularly well suited to be treated with moment methods,
whereby one extrapolates the distribution of the relevant path
functionals on the basis of a few moments. We consider a number of
exotics such as variance knockouts, conditional corridor variance
swaps, gamma swaps and variance swaptions and give valuation
formulas in detail.
\end{abstract}

\maketitle

\section{Introduction}

Volatility derivatives are designed to facilitate the trading of
volatility, thus allowing one to directly take a range of tailored
views. A basic contract is the variance swap which upon expiry pays
the difference between a standard historical estimate of daily
return variance and a fixed rate determined at inception, see
\cite{D1992}, \cite{CM} and \cite{DDKZ}. A variant on this is the
corridor variance swap which differs from the standard variance swap
only in that the underlying's price must be inside a specified
corridor in order for its squared return to be included in the
floating part of the variance swap payout \cite{CL}. A further
generalization is the conditional variance swap which pays the
realized variance of an asset again within some corridor, whereby
the average is taken only over the period when the spot is in the
range. The advantage of conditional corridor variance swaps is that
they allow one to take a view on volatility that is contingent upon
the price level, gaining exposure to volatility only where required.
Although conditional variance swaps appear to be traded, see
~\cite{JP}, this is the first article in the open literature
proposing a pricing methodology.

The numerical method we present extends without difficulties to
other exotic volatility contracts such as variance swaptions, see
also \cite{CLee}. We also take the liberty of inventing exotic
volatility derivatives such as variance knock-out options. These do
not seem to be much traded although perhaps they should be. A
variance knockout can be regarded as a variation on the theme of
barrier knock-out options whereby the knock-out condition is not
triggered by the underlying crossing a certain level. Instead, a
variance knockout vanishes in case realized variance prior to
maturity exceeds a certain pre-assigned threshold. The benefit of a
variance knockout over a barrier knockout is that its hedge ratios
are smoother.

In this paper, we use operator methods for pricing. Our presentation
is self-contained for the specific purpose at hand but see the
review article \cite{A} for a more extended discussion of operator
methods. One point we should stress is that the mathematical and
numerical methods we present would work as efficiently no matter
what underlying model for the stock price process is chosen. We
select one for the purpose of discussing a concrete case and
generating sample graphs, but we are confident the reader can do
better if she intends to refine it. Our mathematical and numerical
methods are model agnostic as they do not rely on closed form
solvability and their performance is not linked to the model
definition. Any model for the stock price dynamics would work just
as well, the only limitation being that market models cannot be
accommodated within the formalism we propose.

We believe that it is important to embed econometric estimates in
the volatility process for the underlying stock price and thus make
it to reproduce realistically the features of the historical
process. Operator methods reviewed in ~\cite{A} are useful in this
respect as they allow one to construct models semi-parametrically or
even non-parametrically while resting assured that numerical
efficiency is not affected by model choice. Our working example is a
3-factor equity model which encompasses stochastic outlook dynamics,
stochastic volatility, local volatility and jumps of both finite and
infinite activity. This process can formally be expressed as
follows:
\begin{equation}
dS_t = \mu_t (S_t)dt + \sigma_t (S_t) dW_t + \mathrm{jumps}.
\label{eq:process}
\end{equation}
Here both the drift and the volatility terms can be freely
specified, thus making available additional degrees of freedom and
allowing for a time-homogeneous calibration.

Armed with the calibrated model we proceed to a moment method which
allows us to obtain conditional moments of integrals of stochastic
processes. Having the first few moments as a starting point, the
distribution is extrapolated by matching moments with
elementary probability distributions functions known in analytically
closed form. On this basis we then price volatility derivatives and
find their hedge ratios.

The article is organized as follows: The next section describes the
underlying model, and the following one the moment method. Sections
(\ref{sec:varswaps}), (\ref{sec:corridor}) and
(\ref{sec:conditional}) describe the variance swaps, corridor
variance swaps and conditional variance swaps respectively, and
their pricing within this framework. In section (\ref{sec:other}) we
look at some more variance related contracts, particularly the gamma
swap and variance knock-out options. A final section concludes.

\section{Underlying Model}
\label{sec:model}

Our working example, as a base model, is specified semi-parametrically
and amounts to a 3-factor equity model with stochastic volatility
and stochastic outlook dynamics. We find that postulating a slowly
varying stochastic outlook dynamics alongside a faster
mean-reverting volatility dynamics is essential to mimic the
historical features of the process and calibrate to both short-dated
and long dated forward volatility. The presence of jumps helps to
explain the shorter maturities, stochastic volatility for the medium
to long term and the outlook dynamics for the long-term features
of the model.

Near time-homogeneity as a key property and requirement for our
model; apart from an exogenously specified term structure of
interest rates, we achieve high precision calibration to all of our
objectives within the bounds of a nearly time-homogenous model. We
stress that time-homogeneity is not a technical constraint from the
numerical point of view, but rather an economic constraint to
discipline the calibration of otherwise very flexible model
specifications. Since parameter sensitivities and hedge ratios one
obtains numerically are remarkably smooth even if one uses single
precision floating point arithmetics on GPUs, the fit to calibration
targets can be made arbitrarily accurate by refining the model with
modest time-inhomogeneities.

The model is specified in terms of a continuous-time Markov chain on
a lattice with finitely many states. We select the following
notations for our discretization: the base lattice $\Lambda$ is made
up by triples $y = (x, a, b)$ where $x = 0 \dots 69$, $a$ takes
values $\texttt{stable}$ and $\texttt{negative}$ is an indicator of
outlook regime and $b$ takes values $\texttt{low}$,
$\texttt{medium}$ and $\texttt{high}$ is an indicator of volatility
regime. The stock price is given by a function $S(x)$. See ~\cite{A}
for the theory behind these discretizations, in particular for
rigorous convergence proofs and sharp bounds on convergence rates.

We specify the Markov generator to precisely express the model
sketched in equation (\ref{eq:process}). This matrix then needs to
be exponentiated numerically to obtain the required transition
probability kernels. We are therefore describing a three-dimensional
problem: one for the underlying stock price, one for the stochastic
volatility regime and one for the stochastic drift regime. Since we
have multiple outlook regimes characterized by different drifts, we
match the jump amplitudes non-parametrically to equal the
risk-neutral drift. We approximate interest rates $r(t)$ as
deterministic and piecewise constant in such a way to match the spot
discount function. With the interest rates taken to be piecewise
constant over time intervals of the order of a quarter, we specify a
Markov generator ${\mathcal{L}}_i$ on each such interval $[t_i,
t_{i+1})$ as having the following form:
\[ {\mathcal{L}}(x, a,b; x', a', b'; t) =
\left( {\mathcal{D}}_{\mathrm{out}}
(x; x' \vert a, b) \delta_{a a'} - {\mathcal{J}}_{\mathrm{out}}
(x, a; x', a'\vert b) \right) \delta_{bb'} \]
\begin{equation}
+ {\mathcal{J}}_{\mathrm{vg}} (x; x' \vert b) \delta_{aa'}
\delta_{bb'} + {\mathcal{V}}_{\mathrm{sv}} (b; b'\vert x, a)
\delta_{xx'} \delta_{aa'}
\end{equation}
The tri-diagonal matrix ${\mathcal{D}}_{\mathrm{out}}
(x; x' \vert a, b) $ is chosen
so that we match the risk neutral drift and satisfy
the no-arbitrage condition in equation (\ref{eq:ArbCond3}) below.
The operator ${\mathcal{J}}_{\mathrm{out}} (x, a; x', a'\vert b)$
models jumps downwards of exponential type which in our
example, are chosen to
be of size $12 \%$ if the outlook is negative and small symmetric
jumps of size $2 \%$ if the outlook is stable. Both jumps are finite
activity and occur with yearly frequency.

The operator ${\mathcal{J}}_{\mathrm{vg}} (x; x' \vert b)$ describes
a variance-gamma (infinite activity) process of log-normal
volatility dependent on the state variable $b$ and of the form
$S^{\beta(S)-1}$ where $\beta(S)$ is the function in Fig.
\ref{figbeta}. The variance rate is chosen to be $4 \%$ and is of
the order of the rate estimated econometrically, not sufficient to
sustain an implied volatility skew over medium or long time
horizons. The first two terms admit the state $x = 0$ as an
absorbing point.

\begin{figure}[htb]
\centering\noindent
        \includegraphics[width=0.5\hsize]{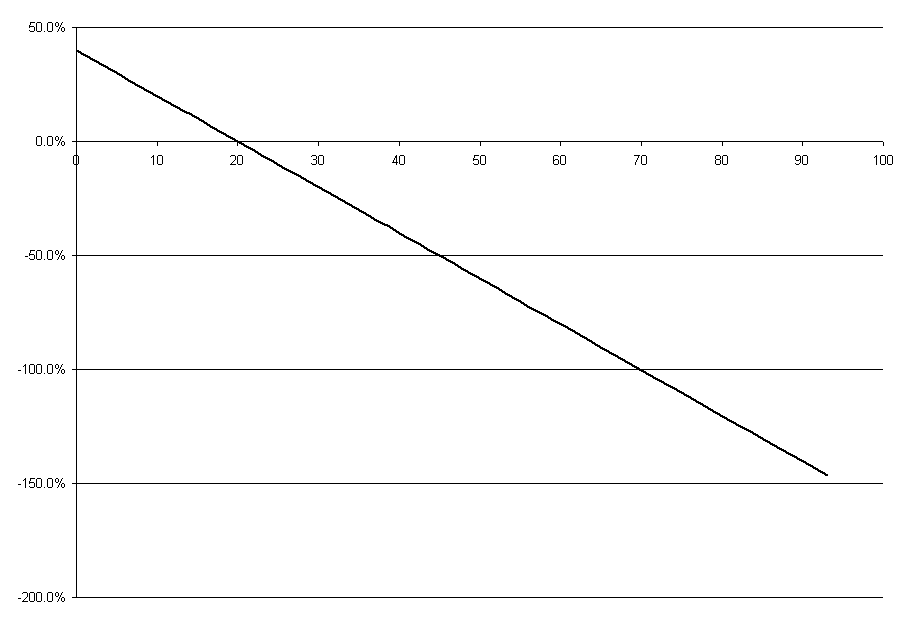}
\caption{Graph of the function $\beta(S)$} \label{varchi}
\label{figbeta}
\end{figure}

The operator ${\mathcal{V}}_{\mathrm{sv}} (b; b'\vert x, a)$
describes a mean reverting process for stochastic volatility given
by volatility regimes. It is allowed to switch between three
possible states: $\texttt{low}$, $\texttt{medium}$ and
$\texttt{high}$. Evidence that the volatility process experiences
patters which can be effectively described by means of regimes has
been previously discussed by Rebonato and Joshi in \cite{RJ}. The
underlying stock price is allowed to transition from one regime to
another only between adjacent regimes.

The Markov generator satisfies the following three conditions
corresponding to probability positivity and conservation and to
arbitrage freedom, respectively:
\begin{equation}
{\mathcal{L}}(x, a,b; x', a', b'; t) \geq 0 \quad \mathrm{if} \quad
(x,a,b) \neq (x',a',b')
\end{equation}
\begin{equation}
\sum_{x',a',b'} {\mathcal{L}}(x, a,b; x', a', b'; t) = 0
\end{equation}
\begin{equation}
\label{eq:ArbCond3}
\sum_{x',a',b'} {\mathcal{L}}(x, a,b; x', a', b'; t)
(S(x') - S(x)) = r(t) S(x)
\end{equation}
The problem of computing the propagator $U(t_1,t_2)$ satisfying the
backward equation
\begin{equation}
\frac{d}{dt_1} U(t_1,t_2) + {\mathcal{L}}(t_1)\, U(t_1,t_2) = 0
\end{equation}
with final time condition $U(t_1,t_2) = \mathbb{I}$, the identity
operator, is solved by exponentiating our generator $\mathcal{L}(t)$
\begin{equation}
U(t_1,t_2) = \exp \left( \int_{t_1}^{t_2} \mathcal{L} (s) ds
\right)
\end{equation}
The matrix exponentiation can be done in a number of ways
~\cite{ML}. The method we advocate is the linear fast exponentiation
algorithm which we find is fast, stable with respect to floating
point errors and can be implemented in single precision on highly
efficient GPU platforms.

Fast exponentiation proceeds as follows: Assume that the generators
$\mathcal{L}(t)$ are piecewise constant as a function of time.
Suppose $\mathcal{L}(t) = \mathcal{L}_{i}$ in the time interval
$[t_i, t_i + (\Delta t)_i]$. Assume $\delta t$ is chosen to be small
enough (but not smaller) that the following two conditions hold:
\begin{equation}
\mathrm{min}_{y \in \Lambda} (1 + \delta t \mathcal{L}_{i} (y,y))
\geq 1/2
\end{equation}
\begin{equation}
\log_2 \frac{(\Delta t)_i}{\delta t} = n \in \mathbb{N}
\end{equation}
This condition leads to intervals $\delta t$ of the order of
one hour of calendar time.
To compute $e^{(\Delta t)_i \mathcal{L}_{i}} (x, y)$,
we first define the elementary propagator
\begin{equation}
u_{\delta t}(x, y) = \delta_{xy} + \delta t \mathcal{L}_i(y, y)
\end{equation}
and then evaluate in sequence
\begin{equation}
u_{2\delta t} = u_{\delta t} \cdot u_{\delta t} \, ,\,
u_{4 \delta t} = u_{2 \delta t} \cdot u_{2 \delta t} \,,\, \dots ,
u_{2^n \delta t} = u_{2^{n-1} \delta t} \cdot u_{2^{n-1} \delta t}
\end{equation}
The matrix multiplication is accomplished numerically by invoking
the level-3 BLAS routine \texttt{sgemm}.

\section{The Moment Method}

\label{sec:moment} One possible strategy to evaluate moments of path
dependent functions is to use Dyson's formula, see \cite{A}. This
formula is sometimes also attributed to Kac ~\cite{DK} and provides
an expression for moments of integrals of stochastic processes.
Moment methods involve reconstructing probability distribution
functions for such integrals by extrapolating out of the knowledge
of the first few moments. This involves selecting a probability
distribution function available in closed form and matching moments.

Consider a time interval $[T,t]$ and a Markov generator
$\L(y_1,y_2; t)$, and consider the integral $I_t$ given by
\begin{equation}
\label{eq:integral1}
I_t = I(y_t,t) \equiv \int_{T}^{t} \phi(y_s,s) ds
\end{equation}

Let's introduce the following one parameter family of deformed
Markov operators parameterized by $\epsilon\in\mathbb{R}$
\begin{equation}
\L_\epsilon(y_1, y_2; t) = \L(y_1, y_2; t) + \epsilon V(y_1, y_2; t)
\end{equation}
where
\begin{equation}
V(y_1, y_2; t) = \phi(y_1; t) \delta_{y_1, y_2} + \L(y_1, y_2; t)
\chi(y_1, y_2; t)(\psi(y_2; t)-\psi(y_1; t))
\end{equation}
The functions $\phi(y;t)$, $\psi(y;t)$ and $\chi(y;t)$ are all continuous
functions. We add the $\psi$ and $\chi$ functions here as we can use them
to model certain path dependents, however in the applications we describe
below we only need the first term $\phi(y;t)$ in order to model the integral
in equation (\ref{eq:integral1}).
\begin{Theorem} {\bf }
For $n=0,1,2,\dots$, we have that
\begin{equation}
\label{eq:MomentFormula1}
\left({d\over d\epsilon}\right)^n\bigg\lvert_{\epsilon=0}
P\exp\bigg(\int_T^{t} \L_\epsilon(s) ds \bigg)(y_1, y_2) = E_T\big[
I_t^n \delta(y_{t} - y_2) \lvert y_{T} = y_1
 \big].
\end{equation}
\end{Theorem}

\begin{proof} Consider Neper's formula for the propagator
\begin{equation}
P \exp\bigg(\int_T^t \L_\epsilon(s) ds \bigg) =
\lim_{N\to\infty} P \prod_{i=1}^N \left( 1 + {t - T\over N} (\L(t_i) +
\epsilon \phi(t_i)) \right)^N
\end{equation}
where $t_i = T + {i t \over N}$.  By collecting similar powers in
$\epsilon$, one finds Dyson's formula
\begin{align}
P \exp\bigg(\int_T^t \L_\epsilon(s) ds \bigg)& = \exp((t-T)\L + \\
& \epsilon \int_T^{t} ds_1
\bigg( e^{(s_1-T)\L} V(s_1) e^{(t-s_1)\L}\bigg) +\\
& \sum_{n=2}^\infty \epsilon^n \int_T^{t} ds_1 ...
\int_{s_{n-1}}^{t} ds_n  \bigg( e^{(s_1-T)\L} V(s_1) e^{(s_2-s_1)\L}
.... V(s_{n}) e^{(t-s_{n})\L}\bigg).
\end{align}
The time-ordered integrals above are proportional to conditional moments, i.e.
\begin{equation}
P\exp\bigg(\int_T^t \L_\epsilon(s) ds \bigg)(y_1, y_2) =
\sum_{n=0}^\infty {\epsilon^n\over n!} E_T\big[ I_t^n
\delta\big(y_{t} - y_2\big) \lvert y_T = y_1\big].
\end{equation}
Here, the factorials originate from the time ordering.
\end{proof}

A technique which is numerically stable in many situations of
practical relevance is to numerically differentiate with respect to
$\epsilon$ the deformed propagators $P \exp\bigg(\int_T^t \L_\epsilon(s) ds \bigg)(y_1, y_2)$
and evaluate at $\epsilon=0$. This technique can be used to obtain
the first moments of $I_t$ on any given bridge for the underlying
Markov process. In most applications, we find that two moments
suffice to extrapolate the probability distribution function to
sufficient accuracy. To do so, it is convenient to choose from among
the probability distribution functions which are analytically
tractable. For instance, starting from the first two moments only,
one can use either the log-normal or the chi-square distribution. In the
above theorem, one only needs to compute the formula for $n$ equal to
the number of moments we want to use.

One may notice that Dyson formula in our Finance application is used
differently from how it is used in Physics. In Physics one usually
has a perturbative method to evaluate moments and the purpose is to
reconstruct the kernel. In Finance, fast exponentiation gives a good
way to evaluate the kernel and Dyson formula is used to evaluate
moments by numerical differentiation.

\section{Contracts on Realized Variance}
\label{sec:varswaps}

Having introduced a stock price model and a moment formula for path
functionals, we next consider volatility derivatives.

A variance swap is a financial contract that upon expiry pays the
difference between a capped measure of the historically realized
variance of daily stock returns and a fixed swap rate determined at
inception. As usual, the swap rate is initially struck at
equilibrium so that the variance swap has zero cost to enter
\cite{CM}.

Variance swaps of maturity $t$ and time of issuance $T$ have a payoff
 given by
\begin{equation}
\label{eq:VarSwapPayoff}
\mathrm{min} \left( \frac{1}{t-T}
\int_{T}^{t} v(y_s,s)ds, f \cdot SR^2 \right) - SR^2
\end{equation}
where $SR$ is the swap rate and $f$ is a factor, a typical value
being $f = 6.2$. Here, $v(y_1, t)$ is the instantaneous variance
defined as follows:
\begin{equation}
v(y_1, t) = \sum_{y_2} {\mathcal{L}}
(y_1,y_2;t) \log^2 \left( \frac{S(y_2)}{S(y_1)} \right)
\end{equation}
The cap is customary on single stocks as, in case of default,
realized log-normal variance may actually diverge. Instead, variance
swaps on indices are usually uncapped and the equilibrium swap rate
is simply given by the first moment
\begin{equation}
\label{eq:VarSwapSimple} \frac{1}{t-T} E_T \left[ \left.
\int_{T}^{t} v(y_s,s)ds \, \delta \big(y_{t} - y_2\big) \right\vert
y_T = y_1 \right] - SR^2.
\end{equation}
The expectation is given by the first derivative in equation
(\ref{eq:MomentFormula1}) without any need for moment matching.
Summing over all final points $y_2$ removes the bridge condition on
the final points.

The payoff of a volatility swap is
\[ \mathrm{min} \left(
\sqrt{\frac{1}{t-T} \int_{T}^{t} v(y_s,s) ds},SR \right) - SR \]
To price these contracts one needs to evaluate the distribution
of realized variance on a bridge, i.e. of the functional
\begin{equation}
RV(y_2) = \delta \left( y_{t} - y_2 \right) \frac{1}{t-T}
\int_{T}^{t} v(y_s,s) ds
\end{equation}
In this case, moment matching is required even in case the cap is
lifted.

\subsection{Matching the moments to a chi-square distribution}
\label{sec:match1}

Let
\[ m_1 = E_T \big[ I_t \left. \delta
\big(y_{t} - y_2\big) \right\vert y_T = y_1 \big]
\quad \quad \mathrm{and} \quad \quad m_2 =
E_T \big[ I_t^2 \left. \delta
\big(y_{t} - y_2\big) \right\vert y_T = y_1 \big] \]
be our first and second pre-assigned moments.
The standard chi-square distribution is given by
\[ f(x) = \frac{1}{2 \Gamma \left( \frac{a}{2} \right) }
\left( \frac{x}{2} \right)^{a/2 - 1} e^{-x/2} \] where $a$ is the
number of degrees of freedom. The first and second
 moments of this distribution are
\[ E[x] = a  \, , \,\,\,\,\,\,\,\,\,\,\,\,\, E[x^2] = a(a+2) \]
To match two pre-assigned moments $m_1$, $m_2$, one can pass to the new variable
\[ \xi = \frac{m_1}{a} x \]
and chose
\[ a = \frac{2m_1^2}{m_2 - m_1^2} \]
The cumulative distribution function is
 \[ F(x;a) = \frac{\gamma \left(\frac{a}{2},\frac{x}{2} \right)}
{\Gamma \left(\frac{a}{2} \right)}, \] where $\Gamma(z)$ and
$\gamma(z,a)$ are the gamma and incomplete gamma functions
respectively, i.e:
\[ \Gamma(z) = \int_{0}^{\infty} s^{z-1} e^{-s} ds \,,\,\,\,\,\,
\gamma(z,a) =  \int_{0}^{a} s^{z-1} e^{-s} ds. \] We find that
\[ E_T \left[ \mathrm{min} (RV,RV_{\mathrm{max}})
\, \delta(y_{t} - y_2)  \right] =
\frac{m_1}{a}  \left[ K(1 - F(K;a)) + a F(K;a +2) \right] \]
and
\[ E_T \left[ \mathrm{min} (\sqrt{RV},\sqrt{RV_{\mathrm{max}}})
\, \delta(y_{t} - y_2)  \right] =
\sqrt{\frac{m_1}{a}} \left[ \sqrt{K} (1- F(K;a))
+ \frac{\sqrt{2} \gamma \left(\frac{a+1}{2},\frac{K}{2} \right)}
{\Gamma \left(\frac{a}{2} \right)} \right] \]
where
\[ a = a(y_1,y_2) = \frac{2m_1 (y_1,y_2)^2}{m_2 (y_1,y_2) - m_1(y_1,y_2)^2}
\,,\,\,\,\, K = K(y_1,y_2) = \frac{a(y_1,y_2)}{m_1(y_1,y_2)}
RV_{\mathrm{max}} \]
These formulas allow one to price both variance swaps and volatility swaps.
Since the dependency on the swap rate in both cases is non-linear,
an exact calculation requires using a root finder. This may be avoided
as an excellent approximation is often obtained by estimating the
variance swap rate as follows:
\begin{equation}
VARSR_T = \sqrt{\sum_{y_2} E_T \lbrack
\mathrm{min} (RV,RV_{\mathrm{max}}) \, \delta(y_{t} - y_2) \rbrack }
\end{equation}
and the volatility swap rate as follows:
\begin{equation}  VOLSR_T = \sum_{y_2} E_T \left[ \mathrm{min} (\sqrt{RV}
,\sqrt{RV_{\mathrm{max}}}) \, \delta(y_{t} - y_2)  \right]
\end{equation}
where
\begin{equation}
\label{eq:RVmax}
RV_{\mathrm{max}} = f \sum_{y_3} U(y_T,T;y_3,t) m_2(y_T,y_3)
\end{equation}

\begin{figure}[htb]
\centering\noindent
        \includegraphics[width=1\hsize]{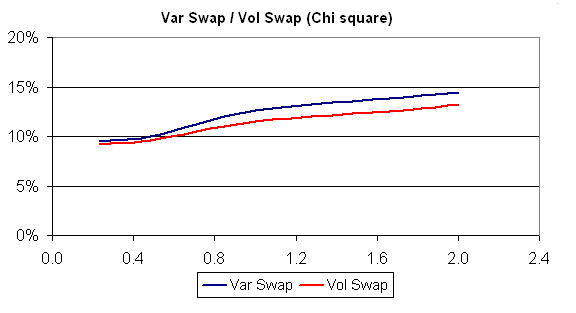}
\caption{Volatility and variance swaps (Chi-square)} \label{varchi}
\end{figure}

\subsection{Matching the moments to a log-normal distribution}
\label{sec:match2}

Let $x$ be log-normally distributed with probability distribution
function
\[ f(x;\mu,\sigma) = \frac{1}{x \sigma \sqrt{2 \pi}}
e^{-(\ln x - \mu)^2 / 2 \sigma^2} \]
then the first two moments are given by
\[ E[x] = e^{\mu + \sigma^2/2} \,\,\,\,\,\,\,\,\,\,\,
\textrm{and} \,\,\,\,\,\,\,\,\,\,\,
E[x^2] = e^{2 \mu + 2 \sigma^2}. \]
Given these expected values, we can obtain relationships for
$\mu$ and $\sigma$ as
\begin{equation}
\label{eq:LogNormalMoments}
\mu = \log \left( \frac{E[x]^2}{\sqrt{E[x^2]}} \right)
\,\,\,\,\,\,\,\,\,\, \textrm{and} \,\,\,\,\,\,\,\,\,\,\,
\sigma^2 = \log \left( \frac{E[x^2]}{E[x]^2} \right)
\end{equation}
and take the pre-assigned moments to be equal to these expected values
$E[x] = m_1$ and $E[x^2] = m_2$. We obtain
\[ E_T [ \mathrm{min} (RV,RV_{\mathrm{max}})
\, \delta(y_{t} - y_2)] \] expressed in terms of the well known
standard normal cumulative distribution function as
\begin{equation} e^{\mu + \sigma^2/2}
{\mathcal{N}} \left( \frac{\log(RV_{\mathrm{max}}) - \mu - \sigma^2/2}
{\sigma} \right)
+ RV_{\mathrm{max}} \, {\mathcal{N}}
\left( \frac{\log(RV_{\mathrm{max}}) - \mu}{\sigma} \right)
\end{equation}
and the second expectation is
\[ E_T \left[ \mathrm{min} (\sqrt{RV},\sqrt{RV_{\mathrm{max}}})
\, \delta(y_{t} - y_2)  \right] = \]
\begin{equation} e^{\frac{1}{8} (4 \mu + \sigma^2)}
{\mathcal{N}} \left( \frac{\log(RV_{\mathrm{max}}) - \mu - \sigma^2/2}
{\sigma} \right) + \sqrt{RV_{\mathrm{max}}}
\, {\mathcal{N}} \left( \frac{\log(RV_{\mathrm{max}}) - \mu}
{\sigma} \right)
\end{equation}
where $\mu$ and $\sigma$ are specified above in terms
of $m_1$ and $m_2$.

\begin{figure}[htb]
\centering\noindent
        \includegraphics[width=1\hsize]{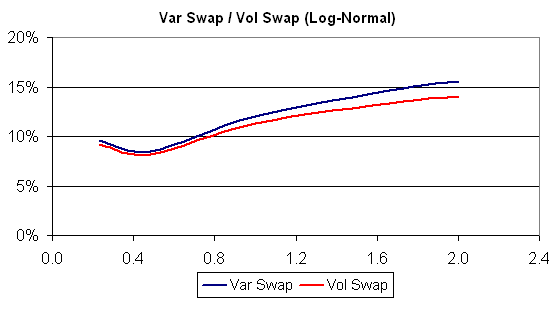}
\caption{Volatility and variance swaps (log-normal)} \label{varln}
\end{figure}

\subsection{Matching higher moments with the Pearson distributions}
\label{sec:match3}

The log-normal and the chi-square distributions allow us to match
the first two moments, here we look at three. Consider the Pearson
Type III distribution which has probability distribution function
\[ f(x) = \frac{1}{b \Gamma(p)} \left( \frac{x - a}{b}
\right)^{p-1} e^{-(x- a)/b} \]
defined on $[a, +\infty)$. The special case of this, when
$a = 0$, $b = 2$ and $p$ is half of an integer, gives the
Chi-Squared distribution. In general, the moments are given by
\begin{equation}
E[x] = a + b p
\end{equation}
\[ E[x^2] = (a + b p)^2 + b^2 p \]
\[ E[x^3] = (a + b p)^3 + 3 b^2 p (a + b p)
+ 2 b^3 p \] and matching these with our pre-assigned moments
$m_1,m_2$ and $m_3$ and computing the values of $a,b$ and $p$ in
terms of these moments we find
\begin{equation}
a = m_1 - \frac{2(m_2 - m_1^2)^2}{m_3 + 2m_1^3 - 3m_1m_2}
\end{equation}
\[ b = \frac{m_3 + 2m_1^3 - 3m_1m_2}{2(m_2 - m_1^2)} \]
\[ p =  \frac{4(m_2 - m_1^2)^3}{(m_3 + 2m_1^3 - 3m_1m_2)^2} \]
Then the expectation in concern, expressed in terms of the
chi-squared cumulative distribution function $F(x;a)$ (section
~\ref{sec:match1}) is
\[ E_T \left[ \mathrm{min} (RV,RV_{\mathrm{max}})
\, \delta(y_{t} - y_2) \right]
= (a+bp - RV_{\mathrm{max}}) F_{\mathrm{Chi}}
\left(2p, 2\frac{RV_{\mathrm{max}}-a}{b} \right) \]
\begin{equation}
- b \left( \frac{RV_{\mathrm{max}}-a}{b} \right)^p
\frac{e^{-\frac{RV_{\mathrm{max}}-a}{b}}}{\Gamma(p)}
+ (2 RV_{\mathrm{max}} - a - bp)
\end{equation}

\section{Corridor Variance Swaps}
\label{sec:corridor}

Corridor variance swaps differ from standard variance swaps in that
the underlying's price must be inside a specified corridor in order
for its squared return to be included in the floating part of the
variance swap payout \cite{CL}.

Corridor variance swap are sometimes engineered so that they afford
full protection against situations where the variance spikes
accompanied by a decline in stock value, but offer only limited
protection in case the volatility surge is associated to a stock
price rally. Following \cite{CL}, one could achieve a robust static
replication by means of a strip of options as in the case of the
variance swap, only that here the hedge is limited to the
pre-specified corridor.
\\
The payoff of corridor variance swaps are given by
\begin{equation}
\frac{1}{t-T} \int_{T}^{t} v(y_s,s) \,\mathbf{1}(L < S(y_s) < H) ds  - SR^2
\end{equation}
so we are interested in the integral
\[ I_t = \frac{1}{t-T} \int_{T}^{t} \phi(y_s,s) ds \]
where
\begin{equation}
\phi(y_1, t) = \sum_{y_2} {\mathcal{L}}
(y_1,y_2;t) \log^2 \left( \frac{S(y_2)}{S(y_1)} \right)
\mathbf{1}(L < S(y_2) < H)
\end{equation}
\begin{figure}[htb]
\centering\noindent
        \includegraphics[width=1\hsize]{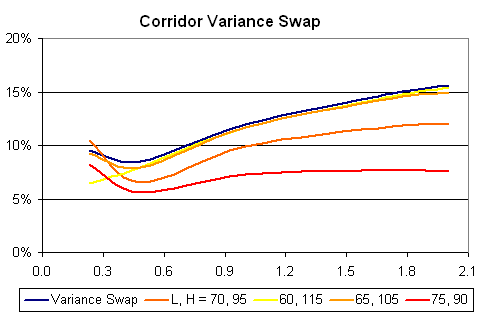}
\caption{Corridor Variance Swaps}
\label{corridor}
\end{figure}
is the "instantaneous \textit{corridor} variance",
and solve the problem of computing its moments by specifying in
this case the operator
\[ {\mathcal{L}}_{\epsilon} (y_1,y_2;t) = {\mathcal{L}} (y_1,y_2;t)
+ \epsilon \phi(y_1,t) \delta_{y_1 y_2} . \] From here the
methodology used in the case of ordinary variance swaps carries
through and one is reduced to computing and matching moments.

\section{Conditional Variance Swaps}
\label{sec:conditional}

Conditional corridor variance swaps are an interesting variant as
they can more easily be interpreted in intuitive terms. In this
case, the payoff is given by the accrued variance divided by the
time elapsed in the range, not the total number of days ~\cite{JP}.
Namely, the payoff is given by
\begin{equation}
\label{eq:CondPayoff}
\frac{\int_{T}^{t} v(y_s,s) \,\mathbf{1}(L < S(y_s) < H) ds}{
\int_{T}^{t} \mathbf{1}(L < S(y_s) < H) ds} - SR^2
\end{equation}
The principal use of conditional derivatives is when
an investor expects a particular asset (e.g. market level,
volatility) to outperform contingent on a level of another asset
being breached. In the case of a conditional variance swap,
the investor may be anticipating that volatility will increase/decrease
if the market were to rally/sell-off or stay within a
range \cite{JP}. While an investor can select the corridor freely, two examples
are the conditional up-variance and down-variance swaps.
Up-variance accrues realised variance only when the
underlying is above a pre-specified level (i.e. no upper barrier),
while down-variance is accrued only when the underlying is
below the specified barrier (i.e. no lower barrier).

To apply the moment method to this case, the first thing to note is
that essentially we are modelling the two integrals appearing in the
payoff at the same time. We are going to need a bi-variate
distribution to do this. Let's write:
\[ I_{t}^{(1)} = \int_{T}^{t} v(y_s,s) \,\mathbf{1}(L < S(y_s) < H) ds \]
\[ I_{t}^{(2)} = \int_{T}^{t} \mathbf{1}(L < S(y_s) < H) ds \]
and in order to compute the expectation
\begin{equation}
\label{eq:Expectation1}
E \left[ \frac{I_{t}^{(1)}}{I_{t}^{(2)}} \right]
\end{equation}
we'll need the following expectations:
\[ E \left[ I_{t}^{(1)} \right] \,,\,\,\,\,
E \left[ I_{t}^{(2)} \right] \,,\,\,\,\,
E \left[ \left( I_{t}^{(1)} \right)^2 \right] \,,\,\,\,\,
E \left[ \left( I_{t}^{(2)} \right)^2 \right] \,,\,\,\,\,
\mathrm{and} \,\,\,\,\,\,
E \left[ I_{t}^{(1)} I_{t}^{(2)} \right] \]
To tackle this problem consider the operator
\begin{equation}
{\mathcal{L}}_{\epsilon_1,\epsilon_2} (y_1,y_2)
= {\mathcal{L}} (y_1,y_2)
+ \epsilon_1 \phi(y_1) \delta_{y_1 y_2}
+ \epsilon_2 \,\psi(y_1) \delta_{y_1 y_2}
\end{equation}
where
\begin{equation}
\phi(y_1, t) = \sum_{y_2} {\mathcal{L}}
(y_1,y_2;t) \log^2 \left( \frac{S(y_2)}{S(y_1)} \right)
\mathbf{1}(L < S(y_1) < H)
\end{equation}
and
\begin{equation}
\psi(y_1) = \mathbf{1}(L < S(y_1) < H)
\end{equation}
then the problem of computing these moments is solved using the
above procedure. In particular, we have
\[ \left. \frac{\partial}{\partial \epsilon_1}
\right\vert_{\epsilon_1 = 0}
e^{((t-T) {\mathcal{L}}_{\epsilon_1,\epsilon_2})}
(y_1,y_2)
= E \left[ \left. I_{t}^{(1)} \right\vert
y_T = y_1, y_{t} = y_2  \right] \]
and
\[ \left. \frac{\partial^2}{\partial \epsilon_1^2}
\right\vert_{\epsilon_1 = 0}
e^{((t-T) {\mathcal{L}}_{\epsilon_1,\epsilon_2})}
(y_1,y_2)
= E \left[ \left. \left( I_{t}^{(1)} \right)^2
\right\vert y_T = y_1, y_t = y_2  \right] \]
similarly (but with respect to $\epsilon_2$) for
\[ E \left[ \left. I_{t}^{(2)} \right\vert
y_T = y_1, y_{t} = y_2 \right]
\,\,\,\,\,\,\, \mathrm{and} \,\,\,\,\,\,\,
E \left[ \left. \left( I_{t}^{(2)} \right)^2
\right\vert y_T = y_1, y_{t} = y_2 \right] \]
The joint expectation will involve the mixed derivative:
\begin{equation}
E \left[ I_{t}^{(1)} I_{t}^{(2)} \right] = \left.
\frac{\partial^2}{\partial \epsilon_1 \partial \epsilon_2}
\right\vert_{\epsilon_1,\epsilon_2 = 0}
e^{((t-T) {\mathcal{L}}_{\epsilon_1,\epsilon_2})}
(y_1,y_2)
\end{equation}
and once computed, we make use of these expectations to match a bivariate
distribution. For simplicity of notation we leave out the conditional part
of these expectations, noting that all the moments we obtain are conditional
to the initial and final points.
\subsection{Bi-Variate log-normal distribution:}

Let
\begin{displaymath}
\left( \begin{array}{c}
Y_1 \\
Y_2 \\
\end{array} \right)
=
\left( \begin{array}{c}
\log(X_1) \\
\log(X_2) \\
\end{array} \right)
\sim N
\left( \left( \begin{array}{c}
\mu_1 \\
\mu_2 \\
\end{array} \right) ,
\left( \begin{array}{c}
\sigma_1 \\
\sigma_2 \\
\end{array} \right) \right)
\end{displaymath}
then $X_1$ and $X_2$ have joint probability distribution function
~\cite{JK}:
\begin{equation}
\label{eq:BivariatePdf}
f(x_1,x_2) = \frac{1}{2 \pi \sigma_1 \sigma_2 \sqrt{1- \rho^2} x_1 x_2}
\cdot
\end{equation}
\[  \exp \left\{ -\frac{1}{2(1- \rho^2)} \left[
\left( \frac{\log x_1 - \mu_1}{\sigma_1} \right)^2
- 2 \rho \frac{(\log x_1 - \mu_1)(\log x_2 - \mu_2)}{\sigma_1 \sigma_2}
+ \left( \frac{\log x_2 - \mu_2}{\sigma_2} \right)^2
\right] \right\} \]
the bivariate log-normal distribution, where $\rho$ is the correlation between
$X_1$ and $X_2$. Both $X_1$ and $X_2$ are log-normally distributed with
\[ E[X_i] = e^{\mu_i + \sigma_i^2/2} \,,\,\,\,\,\,\,\,
\mathrm{and} \,\,\,\,\,\, E[X_i^2] = e^{2 \mu_i + 2 \sigma_i^2}
\,,\,\,\,\, i = 1,2. \] Matching these with the pre-assigned
moments, and solving for $\mu_i$ and $\sigma_i$ we find (similarly
to section \ref{sec:match2})
\[ \mu_i = \log \left(
E \left[ I_{t}^{(i)} \right]^2 \Big/ \,
\sqrt{E \left[\left(I_{t}^{(i)} \right)^2 \right]} \right)
\,\,\,\,\,\,\, \textrm{and} \,\,\,\,\,\,\,\,\,
\sigma_i^2 = \log \left( E \left[ \left( I_{t}^{(i)}\right)^2
\right] \Big/ \, E \left[ \left( I_{t}^{(i)} \right) \right]^2 \right) \]
Moreover, the mixed term is
\[ E \left[X_1 X_2 \right]
= E \left[ e^{Y_1 + Y_2} \right]
= e^{\mu_1 + \mu_2 + \frac{1}{2} (\sigma_1^2
+ \sigma_2^2 + 2 \rho \sigma_1 \sigma_2)} = E[X_1]E[X_2]
\, e^{\rho \sigma_1 \sigma_2} \]
which gives $\rho$ (in terms of the pre-assigned moments):
\begin{equation}
\rho = \frac{1}{\sigma_1 \sigma_2}
\log \frac{E \left[ I_{t}^{(1)} I_{t}^{(2)} \right]}
{E \left[ I_{t}^{(1)} \right] E \left[I_{t}^{(2)} \right]}
\end{equation}

\begin{figure}[htb]
\centering\noindent
        \includegraphics[width=1\hsize]{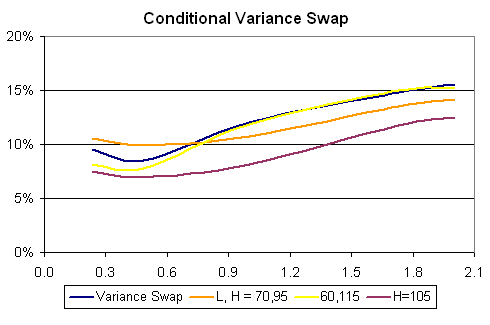}
\caption{Conditional Variance Swaps:
A standard variance swap, two conditional variance swaps of different
corridors, and a 105\% (of spot) conditional up-variance swap.}
\label{conditional}
\end{figure}

Now we are in a position to compute the expectation
$E \left[ I_{t}^{(1)} \Big/ \, I_{t}^{(2)} \right]$.
Given the bivariate log-normal distribution, we have
\[ E \left[ \frac{X_1}{X_2} \right]
= E \left[ e^{Y_1 - Y_2} \right] = E \left[ e^{Y_1 + Y^{'}_2} \right]
\]
where $Y^{'}_2 = -Y_2$ will also be normally distributed
$\sim N(-\mu_2,\sigma_2^2)$, so
\[ E \left[ \frac{X_1}{X_2} \right]
= e^{\mu_1 + \frac{1}{2} \sigma_1^2 } \,
e^{- \mu_2 + \frac{1}{2} \sigma_2^2 } \,
e^{-\rho \sigma_1 \sigma_2} \]
and the expectation (\ref{eq:Expectation1}) is given by
\begin{equation}
E \left[ \frac{I_{t}^{(1)} }{ I_{t}^{(2)} } \right]
= E \left[ I_{t}^{(1)} \right] \cdot
\frac{E \left[ \left( I_{t}^{(2)} \right)^2 \right]}
{E \left[ I_{t}^{(2)} \right]^3} \cdot
\frac{E \left[ I_{t}^{(1)} \right] E \left[I_{t}^{(2)} \right]}
{E \left[ I_{t}^{(1)} I_{t}^{(2)} \right]}
=
\frac{E \left[ I_{t}^{(1)} \right]^2
E \left[ \left( I_{t}^{(2)} \right)^2 \right]}
{E \left[ I_{t}^{(2)} \right]^2
E \left[ I_{t}^{(1)} I_{t}^{(2)} \right] }
\end{equation}
Plugging this in the payoff of the conditional variance swap
(\ref{eq:CondPayoff}) gives us the value of the conditional variance
swap.

We could go a step further and compute the analogy of the capped
variance swap in equation (\ref{eq:VarSwapPayoff}) where we are now
interested in the payoff $E\left[ \mathrm{min} \left(
\frac{X_1}{X_2}, CV_{\mathrm{Max}}  \right) \right]$. Writing this
expectation as a double integral over the bi-variate probability
distribution function (\ref{eq:BivariatePdf}) we find
\begin{equation}
E\left[ \mathrm{min} \left( \frac{X_1}{X_2},
CV_{\mathrm{Max}}  \right) \right]
= \int_{0}^{\infty}  \int_{0}^{X_2 CV_{\mathrm{Max}}}
\left( \frac{X_1}{X_2} - CV_{\mathrm{Max}} \right) f(X_1,X_2)
dX_1 dX_2 + CV_{\mathrm{Max}}
\end{equation}
This double integral will need to be evaluated numerically; one could
apply a Gaussian quadrature.

\section{Gamma Swaps and Variance Knockouts}
\label{sec:other}

In this section, we apply the moment method to other volatility
contracts, namely options on realized variance, variance knock-out
options and gamma swaps.

\subsection{Options on Realized Variance}

Let's consider the example of a call option on realized variance, we
are interested in the payout $E \left[ (RV - K)^{+} \right]$.
Computing the conditional moments for the realized variance as
described above, we can sum over all final points $y_2$ to obtain
the unconditional moments. Matching these with a distribution we
know, as done above, we can look at the payoff
\begin{equation}
E_T\left[ (RV - K)^{+} \right]
\end{equation}
which under the log-normal case is given by the Black-Scholes formula, with
the correct inputs from equation (\ref{eq:LogNormalMoments}).

\subsection{Variance Knock-out Options}

The payoff of a variance knockout call-option with volatility
barrier at $H$ is
\begin{equation}
(D(T)S_T - K)_{+} \mathbf{1} \left(\frac{1}{t-T} \int_{T}^{t}
v(y_s, s)ds < H^2 \right)
\end{equation}
Variance conditional knock-out options are priced as follows:
\begin{equation}
VKO_T = \sum_{y_2} U(y_T,T;y_2,t) \, \mathbf{Pr}
\left(RV(y_2) < H^2 \right)
\left( D(T')S(y_2,t) - K \right)_+
\end{equation}
where $\mathbf{Pr} \left(RV(y_2) < H^2 \right)$, the probability
that $RV(y_2) < H^2$ is obviously the cumulative distribution function
(for the distribution which we chose to model $RV$ with).

\subsection{Gamma swaps}

The gamma swap (also known as weighted variance swap) is similar to
the variance swap only we have an additional spot factor. They
reduce the weighting of the downside scenarios but introduce a
market level exposure. The gamma variance swap pays at maturity the
realized weighted variance of an asset
\[ \frac{252}{m} \sum_{i}^{n} \log
\left( \frac{S_{i}}{S_{i-1}} \right)^2 \frac{S_{i}}{S_{0}} \]
We note the following features of the Gamma swap
\\
(i) As the spot approaches zero, we do not have a blow up as in the
case of the variance swap: $\lim_{x \to 0} x \log(x) = 0$,
\\
(ii) Similar to the variance swap case,
one can split the gamma swap payoff
into a European payoff and a daily hedging strategy. However in the case
of the Gamma swap, the weights in the replicating portfolio are of order
$1/K$ compared to those in the variance swap of order $1/K^2$.

To price the gamma swap within our framework, we look at the
instantaneous level: we no longer have the instantaneous
variance, but the "instantaneous gamma" given by
\[ g(y_1, t) = \sum_{y_2} {\mathcal{L}}
(y_1,y_2;t) \log^2 \left( \frac{S(y_2)}{S(y_1)} \right)
\frac{S(y_2)}{S(y_1)} \]
Gamma swaps of maturity $t$ and time of issuance $T$ have a payoff
given by
\begin{equation}
\frac{1}{t-T} \int_{T}^{t} g(y_s,s)ds - SR^2
\end{equation}

\begin{figure}[htb]
\centering\noindent
        \includegraphics[width=1\hsize]{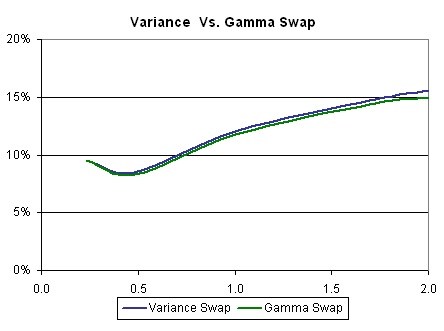}
\caption{Gamma swap}
\label{gamma}
\end{figure}

\section{Summary and extensions}

Using operator theory we presented a 3-factor equity model which is
semi-parametric or even non-parametric and allows for a
time-homogeneous calibration. A moment method based on Dyson's
formula is then given and used to compute the moments of integrals
of the underlying Markov process. We then apply this method to a
various volatility derivatives which we find to be well suited to
moment methods. We consider a number of exotics such as variance
knockouts, conditional corridor variance swaps, variance swaptions
and gamma swaps. Straightforward extensions would include for
instance conditional corridor variance knockout options, forward
starting variance swaptions, etc..

\bibliographystyle{giwi}
\bibliography{moments}

\end{document}